\documentclass[a4paper,11pt]{article}

\pdfoutput=1
\usepackage{microtype}
\usepackage{amssymb}

\usepackage{graphicx}

\usepackage{multirow}
\usepackage{multicol}
\usepackage{theorem,latexsym}
\usepackage{amsmath,amssymb,enumerate}
\usepackage{xspace}

\newtheorem{theorem}{Theorem}[section]

\newtheorem{lemma}[theorem]{Lemma}

\newtheorem{definition}[theorem]{Definition}

\newtheorem{claim}[theorem]{Claim}
\newtheorem{invariant}[theorem]{Invariant}

\newenvironment{proof}{\vspace{-0.05in}\noindent{\bf Proof:}}%
        {\hspace*{\fill}$\Box$}

        {\hspace*{\fill}$\Box$\par}

\newcommand{\del}{\delta}
\newcommand{\eps}{\epsilon}
\newcommand{\oeps}{(1 + \eps)}
\newcommand{\sm}{\setminus}

\newcommand{\ot}{\tilde{O}}

\newcommand{\ceil}[1]{\lceil {#1} \rceil}
\newcommand{\floor}[1]{\lfloor {#1} \rfloor}

\newcommand{\old}{\mbox{\sc old}}
\newcommand{\round}[2]{\mbox{{\sc round}}_{#1}(#2)}
\newcommand{\indep}{\mathcal{IV}}
\newcommand{\neighbors}{\mbox{\sc neighbors}}
\newcommand{\parent}{p}
\newcommand{\level}{\mbox{\bf level}}
\newcommand{\dist}{\mbox{\bf dist}}
\newcommand{\cutoff}{\mbox{\bf cut-off}}
\newcommand{\bound}[1]{\mbox{{\sc bound}}_{#1}}
\newcommand{\wses}{\mbox{\bf WSES}}
\newcommand{\es}{\mbox{\bf ES}}
\newcommand{\oweight}{w_o}
\newcommand{\olevel}{\level_o}
\newcommand{\firstdif}{{\sc first-difference}}
\newcommand{\seconddif}{{\sc second-difference}}
\newcommand{\thirddif}{{\sc third-difference}}

\newcommand{\matha}{\mathcal{A}}

\newcommand{\gtau}{G_{\tau}}
\newcommand{\vtau}{V_{\tau}}
\newcommand{\etau}{E_{\tau}}
\newcommand{\dtau}{\dist_{\tau}}
\newcommand{\pitau}{\pi_{\tau}}
\newcommand{\dptau}{\dist'_{\tau}}
\newcommand{\ds}{d^{s}}
\newcommand{\gtaus}{\gtau^s}

\newcommand{\heavy}{\mbox{\sc heavy}}
\newcommand{\light}{\mbox{\sc light}}

\newcommand{\heavytau}{\heavy(\tau)}
\newcommand{\lighttau}{\light(\tau)}
\newcommand{\ghtau}{G[\heavytau]}

\newcommand{\vstarh}{V^*_H}
\newcommand{\gstar}{G^*}
\newcommand{\estar}{E^*}
\newcommand{\pistar}{\pi^*}
\newcommand{\dstar}{\dist^*}

\newcommand{\ballone}{\mbox{\sc ball}}

\title{Deterministic Partially Dynamic Single Source Shortest Paths in Weighted Graphs}

\author{
 Aaron Bernstein\thanks{Department of Mathematics, Technical University of Berlin. Supported in part by the Einstein Grant at TU Berlin. \tt{bernstei@gmail.com}}}

\begin{document}

\maketitle

\begin{abstract}
In this paper we consider the decremental single-source shortest paths (SSSP) problem, where given a graph $G$ and a source node $s$ the goal is to maintain shortest distances between $s$ and all other nodes in $G$ under a sequence of online adversarial edge deletions.
In their seminal work, Even and Shiloach [JACM 1981] presented an exact solution to the problem in unweighted graphs with only $O(mn)$ total update time over all edge deletions. 
Their classic algorithm was the state of the art for the decremental SSSP problem for three decades, even when approximate shortest paths are allowed.

The first improvement over the Even-Shiloach algorithm was given by Bernstein and Roditty [SODA 2011], 
who for the case of an unweighted and undirected graph presented a $(1+\epsilon)$-approximate  algorithm with constant query time and a total update time of
$O(n^{2+o(1)})$.
This work triggered a series of new results, culminating in a recent breakthrough of Henzinger, Krinninger and Nanongkai [FOCS 14], who presented a $(1+\epsilon)$-approximate algorithm for
undirected weighted graphs whose total update time is near linear: $O(m^{1+o(1)}\log(W))$,
where $W$ is the ratio of the heaviest to the lightest edge weight in the graph.
In this paper they posed as a major open problem the question of derandomizing their result.

Until very recently, all known improvements over the Even-Shiloach algorithm were randomized
and required the assumption of a non-adaptive adversary.
In STOC 2016, Bernstein and Chechik showed the first \emph{deterministic} 
algorithm to go beyond $O(mn)$ total update time: 
the algorithm is also $(1+\epsilon)$-approximate, and has total update time $\tilde{O}(n^2)$.
In SODA 2017, the same authors presented an algorithm with total update time $\tilde{O}(mn^{3/4})$.
However, both algorithms are restricted to undirected, unweighted graphs.
We present the \emph{first} deterministic algorithm for \emph{weighted} undirected graphs to go beyond the $O(mn)$ bound. 
The total update time is $\tilde{O}(n^2 \log(W))$. 

\end{abstract}

\section{Introduction}

The objective of dynamic graph algorithms is to handle an online sequence of update operations while maintaining a desirable functionality on the graph, e.g., the ability to answer shortest path queries.
An update operation may involve a deletion or insertion of an edge or a node, or a change in an edge's weight.
In case the algorithm can handle only deletions and weight increases it is called decremental,
if it can handle only insertions and weight decreases it is called incremental, and if it can handle both it is called fully dynamic.

In this paper we consider the problem of (approximate) single source shortest paths (SSSP) in unweighted undirected graphs, in the decremental setting.
Specifically, given an undirected graph $G$ with positive weights and a source node $s$, our algorithm needs to preform the following operations:
1) Delete($e$) -- delete the edge $e$ from the graph 
2) Increase-Weight($e$) -- increase the weight of $e$
3) Distance($v$) -- return the distance between $s$ and $v$, i.e., $\dist(s,v)$, in the current graph $G$.

Fully dynamic shortest paths has a very clear motivation, as computing shortest paths in a graph is one of the fundamental problems of graph algorithms, and many shortest path applications must deal with a graph that is changing over time. The incremental setting is somewhat more restricted, but is applicable to any setting in which the network is only expanding. The decremental setting is often very important from a theoretical perspective, as decremental shortest paths (and decremental \emph{single source} shortest paths especially) are used as a building block in a very large variety of fully dynamic shortest paths algorithms; see e.g. \cite{HeKi01,King99,Be09,BaHaSe03,RoZw12,AbrahamCT14} to name just a few. Decremental shortest paths can also have applications to non-dynamic graph problems;
see e.g. Madry's paper on efficiently computing multicommodity flows \cite{Madry10}.

We say that an algorithm has an approximation guarantee of $\alpha$ if its output to the query Distance($v$) is never smaller than the actual shortest distance and is not more than $\alpha$ times the shortest distance.
Dynamic algorithms are typically judged by two parameters:
the time it takes the algorithm to adapt to an update (the Delete or Increase-Weight operation),
and the time to process a query (the Distance operation).
Typically one tries to keep the query time small
(polylog or constant),
while getting the update time as low as possible.
All the algorithms discussed in this paper have constant query time, unless noted otherwise.
In the decremental setting, which is the focus of this paper,
one usually considers the aggregate sum of update times
over the \emph{entire} sequence of deletions,
which is referred to as the {\em total update time}.

\paragraph{Related Work:}
The most naive solution to dynamic SSSP is to simply invoke a static SSSP algorithm after every deletion, which requires $\tilde{O}(m)$ time, using
e.g. Dijkstra's algorithm. (The $\tilde{O}$ notation suppresses polylogarithmic factors.)
For unweighted graphs, since there can be a total of $m$ deletions,
the total update time for the naive implementation is thus $\Omega(m^2)$.

For fully dynamic SSSP, nothing better than the trivial $O(m)$ time per update is known. That is,
we do not know how to do better than reconstructing from scratch after every edge update. For this reason,
researchers have turned to the decremental (or incremental) case in search of a better solution.
The first improvement stems all the way back to 1981, when Even and Shiloach \cite{EvenS81}
showed how to achieve total update time $O(mn)$ in unweighted undirected graphs.
A similar result was independently found by Dinitz \cite{Di06}.
This was later generalized to directed graphs by Henzinger and King \cite{HenzingerK95FOCS}.
This $O(mn)$ total update time bound is still the state of art, and there are conditional lower bounds
\cite{RoZw04b,HenzingerKNS15} showing that it is in fact optimal up to log factors. (The reductions are to boolean matrix multiplication and the online matrix-vector conjecture respectively). 

These lower bounds motivated the study of the approximate version of this problem.
In 2011, Bernstein and Roditty \cite{BernsteinR11} presented the first algorithm to go beyond the $O(mn)$
bound of Even and Shiloach \cite{EvenS81}:
they presented a $(1+\eps)$-approximate decremental SSSP algorithm for undirected unweighted graphs with
$O(n^{2+O(1/\sqrt{\log{n}})}) = O(n^{2+o(1)})$ total update time.
Henzinger, Krinninger and Nanongkai \cite{HeKrNa14} later improved the total update time to
$O(n^{1.8+o(1)} + m^{1+o(1)})$,
and soon after the same authors \cite{HenzingerKNFOCS14} achieved a close to optimal total update time of
$O(m^{1 + o(1)} \log{W})$ in undirected weighted graphs, where $W$ is the largest weight in the graph (assuming the minimum edge weight is 1).
The same authors also showed that one can go beyond $O(mn)$ total update time bound in \emph{directed} graphs (with a $\oeps$ approximation) \cite{HenzingerKNSTOC14,HenzingerKN15}, although the state of art is still only a small improvement: total update time $O(m n^{0.9 + o(1)} \log{W})$.

However, every single one of these improvements over the $O(mn)$ bound relies on randomization, and
has to make the additional assumption of a \emph{non-adaptive} adversary. In particular,
they all assume that the updates of the adversary are completely independent from the shortest paths
or distances returned to the user, i.e. that the updates are fixed in advance.
This makes these algorithms unsuitable for many settings,
and also prevents us from using them as a black box data structure. 
For this reason, it is highly desirable to have deterministic algorithms for the problem.

Very recently, in STOC 2016, Bernstein and Chechik presented the first \emph{deterministic} algorithm to go beyond $O(mn)$ total update time, again with a $\oeps$ approximation. The total update time is $\ot(n^2)$ \cite{BernsteinC16}. They followed this up with a second deterministic algorithm \cite{BernsteinC17} that has total update time $\ot(n^{1.5}\sqrt{m}) = \ot(mn^{3/4})$. 
Both algorithms rely on the same basic technique, so this is currently the only know technique for deterministically breaking through the $O(mn)$ barrier.
However, both results above were limited to undirected, unweighted graphs. In this paper, we show that this core technique
can also be applied to weighted graphs.

\paragraph{Our Results:}
\begin{theorem}
\label{thm:main}
Let $G$ be an undirected graph with positive edge weights
subject to a sequence of edge deletions and weight increases,
let $s$ be a fixed source, and let $W$ be the ratio between the largest and smallest edge weights in the graph.
there exists a \emph{deterministic} algorithm that maintains $\oeps$-approximate
distances from $s$ to every vertex
in total update time $O(m\log^2(n)\log(nW) + n^2\log(n)\log(nW)\eps^{-2})$.
The query time is $O(1)$.
(Like most decremental shortest paths algorithms, our result can very easily be extended with the same update time to the incremental case, where the update sequence contains edge insertions and weight increases.)
\end{theorem}

(Technical Note: In weighted graphs the number of updates can be very large 
-- i.e. if each update just increases some edge weight by some small $\eps$ --
and so in addition to the total update time above, the algorithm necessarily requires
an additional $O(1)$ per update. This $O(1)$ factor is present in all dynamic algorithms,
and is typically omitted.)

Our algorithm does not match the randomized state of the art of $\ot(m^{1 + o(1)}\log(W))$
\cite{HenzingerKNFOCS14},
but it is optimal up to log factors for dense graphs,
and is the first deterministic algorithm to go beyond the $O(mn)$ barrier in weighted graphs.
The algorithm is also much simpler than the randomized algorithm for weighted graphs 
, and it does not incur the extra $n^{o(1)}$ factor present in the randomized algorithms.
Thus, our algorithm is in fact faster than \emph{all} existing randomized algorithms for the problem
in dense weighted graphs where $m = \Omega(n^{2 - 1/\sqrt{\log(n)}})$.

As a final remark, we note that all previous deterministic algorithms to 
go beyond the $O(mn)$ bound \cite{BernsteinC16, BernsteinC17} have the strange
drawback that there is no obvious way to return an approximate path,
only a distance. Our algorithm unfortunately shares this drawback.
The reason is that whereas in other dynamic shortest path algorithms
the distance returned corresponds to some path in the graph,
our distance involves an additive error that is bounded through
a structural claim about the non-existence of certain paths,
but does not itself correspond to an actual path 
(see Lemma 4.4 in \cite{BernsteinC16}, Lemma 5.3 in \cite{BernsteinC17}, and 
Lemma \ref{lem:gtau-approx} in our paper.) 

\paragraph{Preliminaries:}
In our model, an undirected weighted graph is subject to deletions and weight increases.
All weights in the original graph are positive.
Let $G = (V,E)$ always refer to the \emph{current} versions of the graph.
Let $m$ refer to the number of edges in the original graph, and
$n$ to the number of vertices.
Let $w(u,v)$ refer to the weight of edge $(u,v)$.
\emph{Let us assume for simplicity that the minimum weighted
in the original graph is at least 1},
and note that since edge weights only increase, 
this will be true of all version of the graph.
Let $W$ be the largest edge weight to ever appear in the graph.
For any pair of vertices $u,v$,
let $\pi(u,v)$ be the shortest $u-v$ path in $G$ (ties can be broken arbitrarily),
and let $\dist(u,v)$ be the length of $\pi(u,v)$.
Let $s$ be the fixed source from which our algorithm must maintain approximate distances,
and let $\eps$ refer to our approximation parameter;
when the adversary queries the distance to a vertex $v$,
the algorithm's guarantee is to return a distance $\dist'(v)$
such that $\dist(s,v) \leq \dist'(s,v) \leq (1 + O(\eps)) \dist(s,v)$.
Note that the $(1 + O(\eps))$ approximation factor can always be reduced
to $(1 + \eps)$ without affecting the asymptotic running time
by simply starting with a suitably smaller $\eps'$.

Given any two sets $S, T$ we define the set difference
$S \sm T$ to contain all elements $s$ such that $s \in S$ but $s \notin T$.
We will measure the update time of the dynamic subroutines used by our algorithm
in terms of their \emph{total} update time over the entire sequence of edge changes.
Note that although edges in the main graph $G$ are only being deleted,
there may be edge insertions into the auxiliary graphs used by the algorithm.

\section{High Level Overview}
Our algorithm extends the techniques in the STOC 2016 paper of Bernstein and Chechik 
\cite{BernsteinC16} from unweighted to weighted graphs. We now briefly summarize
their approach for unweighted graphs, aiming in this overview for total update time
$\ot(n^{2.5})$ instead of $\ot(n^2)$. We start with the well known fact that the classic Even and Shiloach algorithm 
has total update time $O(n^{2.5})$ \cite{EvenS81} if we only care about distances less than
$\sqrt{n}$. So the hard case is long distances in a dense graph. The key
observation of \cite{BernsteinC16} is that these two problematic poles cannot coexist.
In particular, Bernstein and Chechik showed that any shortest path contains
at most $3\sqrt{n}$ vertices of degree more than $\sqrt{n}$. The basic reason
is that if we look at every third high-degree 
vertex on the shortest path, then these vertices must have mutually disjoint neighborhoods,
since otherwise there would be a shorter path between them of length 2;
the claim then follows because each high-degree vertex has at least $\sqrt{n}$ neighbors,
and there are only $n$ vertices in total. 
Thus, Bernstein and Chechik show that the total 
contribution of high-degree vertices to the shortest path is small.
They then argue that this allows us to effectively "ignore" 
all vertices of degree $\geq \sqrt{n}$, at the cost of an additive error of only $O(\sqrt{n})$.
(we do not go into details of this ignoring here.)
This completed their result because for distances less than $O(\sqrt{n}\eps^{-1})$
the classic Even and Shiloach has total update time $O(n^{2.5}\eps^{-1})$, whereas
for longer distances we can afford the $O(\sqrt{n})$ additive error from ignoring
high degree vertices (it is subsumed into the $\oeps$ multiplicative error), 
so the resulting graph only $O(n^{1.5})$ edges.

This technique has no easy extension to weighted graphs via scaling.
To see this, say there were two edge weights, $1$ and $\sqrt{n}$, and that
every vertex had very few incident edges of weight $1$, but $\sqrt{n}$
edges of weight $\sqrt{n}$. Then it is easy to see that distances in the graph
could still be as large as $n$ (as opposed to $\sqrt{n}$ in the unweighted case). 
But we cannot scale weights down, or use some
sort of hop-distance technique found elsewhere in the literature, because although
the heavy edges are the ones that cause the graph to be dense, the shortest path
could still contain many edges of weight $1$. There is thus no way in a weighted graph to ignore high-degree vertices.

To overcome this, we switch focus to a notion of degree that considers edge weights. 
In particular, note that in 
the example above, although the graph was dense,
the total number of edges of weight $1$ was small. 
The first step in our result is to formalize this observation, and show that 
although we cannot simply ignore vertices of high degree, for any vertex $v$ we
\emph{can} ignore its low-weight incident edges if there are many of them.
The proof of this is similar to the proof that we can afford to 
ignore dense vertices in unweighted graphs,
but because we have to deal with many edge different weights at once, 
the analysis is more sophisticated.

The problem is that on its own ignoring large neighborhoods of low-weight edges doesn't help: 
the result is a graph that is guaranteed to only have a small number of low-weight edges
(and a medium number of medium-weight edges), but the graph can still be very dense with
high-weight edges. To overcome this, we develop a variation on the standard Even and Shiloach
algorithm \cite{EvenS81} which is more efficient at dealing with heavy edges
but incurs a $\oeps$ approximation.
The basic intuition is as follows, though the details are somewhat complex. The classic algorithm
guarantees that we only touch an edge $(u,v)$ when the distance $\dist(s,u)$ or $\dist(s,v)$ increases.
But if $(u,v)$ has high weight, and $\dist(s,u)$ only increases by $1$, then the distance
increase is insignificant with respect to the total weight on $(u,v)$, and so we can 
afford to ignore it. 
The number of times we touch edge $(u,v)$ thus ends up being proportional to 
$\eps^{-1}/w(u,v)$.

All in all, our algorithm is based on the framework of Bernstein and Chechik
for unweighted graphs \cite{BernsteinC16}, but requires significant modifications
both to the high-level orientation (we need a new notion of sparsity that depends on edge weights),
as well as to the technical details. The basic approach introduced in \cite{BernsteinC16} of ignoring certain classes of dense vertices has
proved quite versatile and powerful 
(see the two very different applications of it in 
\cite{BernsteinC16} and \cite{BernsteinC17}), and is currently
the \emph{only} known approach for going beyond $O(mn)$ deterministically, so 
our extension to weighted graphs might be useful in further applications
of this approach. Also, our extension of the classic Even and Shiloach algorithm is
presented in extreme generality, and might prove useful in other dynamic shortest path problems
where one has a graph that has more edges of high weight than of low weight.

\section{The Threshold Graph}

In our algorithm, each vertex will have a weight cutoff, denoted $\cutoff(v)$,
and will end up "ignoring" the edges below that cutoff. As described in the high-level
overview above, the required degree to ignore edges goes up as the edge weights go up.

\begin{definition}
\label{def:threshold}
Let the level of edge $(u,v)$ be $\level(u,v) = \left\lfloor \log(w(u,v)) \right\rfloor$.
For any vertex $v$, let $I(v)$ contain all edges incident to $v$ 
($I$ for incident), and let $I_i(v)$ contain all edges incident to $v$
of level $i$ or less. Now, for any positive threshold $\tau$ (not necessarily integral), 
and any vertex $v$,
we define $\cutoff_{\tau}(v)$ to be the \emph{maximum}
index $i$ such that $I_i(v)$ contains at least $\tau 2^i$ edges;
if no such index exists, $\cutoff_{\tau(v)}$ is set to $0$.
We say that an edge $(u,v)$ is $\tau$-heavy if
$\level(u,v) \leq \cutoff_\tau(u)$ OR $\level(u,v) \leq \cutoff_\tau(v)$;
we say that it is $\tau$-light otherwise. 
Let $\heavytau$ be the set of all $\tau$-heavy edges in $G$,
and let $\lighttau$ be the set of all $\tau$-light edges in $G$. 
Let $\ghtau$ be the graph $(V,\heavytau)$.
Note that the levels, cut-offs, heaviness, and lightness are
\emph{always} defined with respect to the main graph $G$, 
never with respect to auxiliary graphs. 
Since $\tau$ is usually clear from context, we often just write $\cutoff(v)$, heavy, and light.
\end{definition} 

\begin{definition}
\label{def:gtau}
Given a graph $G = (V,E)$ with $|V| = n$ and $|E| = m$, and a 
positive threshold $\tau$,
define the threshold graph $\gtau = (\vtau, \etau)$ as follows
\begin{itemize}
\item $\vtau$ contains every vertex $v \in V$.
\item $\vtau$ also contains an additional vertex $c$ for each
connected component $C$ in the graph $\ghtau = (V, \heavytau)$
\item $\etau$ contains all edges in $\lighttau$, with their original weights.
\item For every vertex $v \in V$, $\etau$ contains an
edge from $v$ to $c$ of weight $1/2$, where $c$ is the component vertex in
$\vtau \sm V$ that corresponds to the component $C$ in $\ghtau$ that contains $v$.
\end{itemize}
For any pair of vertices $s,t \in V$ define $\pitau(s,t)$ to be the shortest
path from $s$ to $t$ in $\gtau$, and define $\dtau(s,t)$ to be the weight of this path.
\end{definition}

\begin{lemma}
\label{lem:gtau-sparse}
For any $i$, $\gtau$ contains at most $n$ edges of weight $1/2$ that are not in $E$ 
(one per vertex in $V$),
as well at most $O(n \cdot \tau \cdot 2^i)$ edges at level $i$.
\end{lemma}

\begin{proof}
An edge $(u,v) \in E$ of level $i$ is only in $\etau$ 
if $(u,v)$ is $\tau$-light, i.e. 
if $\cutoff(u) < i$ AND $\cutoff(v) < i$. 
But a vertex with $\cutoff$ less than $i$ has by definition less than
$2^i \tau$ incident edges of level $i$. 
\end{proof}

The following Lemma is the weighted analogue of Lemma 4.4 in \cite{BernsteinC16},
but the most direct extension of that proof to the weighted case would
analyze each of the $\log(W)$ edge levels separately, and thus 
multiply the error by $\log(W)$, eventually leading to a $\log^3(W)$ factor
in the total running time of the algorithm. To avoid this extra error, we need a somewhat
more sophisticated analysis than in Lemma 4.4 of \cite{BernsteinC16}

\begin{lemma}
\label{lem:gtau-approx}
For any graph $G = (V,E)$, any positive integer threshold $\tau$,
and any pair of vertices $s,t \in V$:
$$\dtau(s,t) \leq \dist(s,t) < \dtau(s,t) + \frac{14n}{\tau}$$.
\end{lemma}

\newcommand{\lpi}{L_{\pitau}}

\begin{proof}
It is not hard to see that we have $\dtau(s,t) \leq \dist(s,t)$.
Simply consider the shortest $s-t$ path $\pi(s,t) \in G$.
All the edges in $\pi(s,t) \bigcap \lighttau$ are also in $\gtau$.
Otherwise, for any $\tau$-heavy edge $(u,v)$ on $\pi(s,t)$,
recall that the weight $w(u,v)$ in $G$ is always at least $1$,
and note that because edge $(u,v)$ is $\tau$-heavy, the vertices $u$ and $v$ must
be in the same connected component in $\ghtau$,
and so there is a path of length $1$ from $u$ to $v$ in $\gtau$:
the edge of weight $1/2$ from $u$ to the component vertex $c$,
and then a second edge from $c$ to $v$.

We now prove that $\dist(s,t) \leq \dtau(s,t) + \frac{14n}{\tau}$.
Let $\pitau(s,t)$ be the shortest $s-t$ path in $\gtau$.
Consider the edge set $\estar \subseteq E$ which contains all the $\tau$-light
edges of $\pitau(s,t)$ (these edges are also in $E$), and all $\tau$-heavy edges in $E$.
That is, $\estar = \heavytau \bigcup (\lighttau \bigcap \pitau(s,t))$.
Let $\gstar = (V, \estar)$.
We first observe that there exists an $s-t$ path in $\gstar$.
We construct this path by following $\pitau(s,t)$.
$\pitau(s,t)$ contains $\tau$-light edges, which are by definition in $\gstar$, 
as well as subpaths of length $2$ of the form $(v,c) \circ (c,w)$,
where $v$ and $w$ are in the same connected component $C$ in $\ghtau$.
But since $v$ and $w$ are in the same connected component,
there is a path of only heavy edges connecting them,
and all heavy edges are in $\gstar$.

Now, let $\pistar(s,t)$ be the \emph{shortest} $s-t$ path in $\gstar$.
Let $\dstar(s,t)$ be the length of $\pistar(s,t)$.
Since $\gstar$ is a subgraph of $G$, we know that
$\dist(s,t) \leq \dstar(s,t)$. We now show that
\begin{equation}
\label{eq:dstar-upper}
\dstar(s,t) < \dtau(s,t) + \frac{14n}{\tau}
\end{equation}
which will complete the proof of the lemma.

We prove Equation \ref{eq:dstar-upper} by showing that
all the heavy edges on $\pistar(s,t)$ have total weight at most $\frac{14n}{\tau}$:
since all the light edges on $\pistar(s,t)$ are by definition also on $\pitau(s,t)$,
this proves the equation.

Let $\vstarh$ be the set of vertices on $\pistar(s,t)$ that
are incident to at least one heavy edge on $\pistar(s,t)$.
Now, for any vertex $v \in \vstarh$,
let $\ballone(v)$ contain $v$ and all vertices $u \in V$, such that
there is an edge $(u,v)$ with $\level(u,v) \leq \cutoff(v)$.
Note that by definition of $\cutoff(v)$, we have $|\ballone(v)| \geq 2^{\cutoff(v)} \tau$. 

Now, let $P$ be an ordering of the vertices in $\pistar(s,t)$ in non-increasing order
of $\cutoff$: so if $i < j$ then $\cutoff(P[i]) \geq \cutoff(P[j])$. 
We will now construct a set $\indep$ of \emph{independent} vertices in $\vstarh$
as follows. We start with $\indep$ empty. Now, go through the vertices in $\vstarh$ according to their order in $P$ (so in non-increasing order of cut-off),
and for each $v$ do the following: if $\ballone(v)$ is disjoint
from $\ballone(u)$ for every $u \in \indep$, add $v$ to $\indep$.
Clearly if $v$ and $w$ are independent 
then $\ballone(v)$ and $\ballone(w)$ are disjoint.
But then recalling that for any $v$, $|\ballone(v)| \geq 2^{\cutoff(v)}\tau$,
and that the total number of vertices in the graph is $n$,
we have

\begin{equation}
\label{eq:sum-cutoffs}
\sum_{v \in \indep} 2^{\cutoff(v)} \leq n/\tau
\end{equation}

We will end up showing that each independent vertex $v$ 
is "responsible" for at most $14 \cdot 2^{\cutoff(v)}$ weight
from heavy edges on $\pistar(s,t)$, which will complete our proof.
We say that a non-independent vertex $u$ belongs to independent vertex $v$
if $\cutoff(v) \geq \cutoff(u)$ and $\ballone(u)$ and $\ballone(v)$ are non-disjoint. 
By our independence construction, every non-independent vertex $u$ 
belongs to one or more independent vertices.

\begin{claim}
\label{claim:belonging-distance}
If $v$ is independent,
and $u$ belongs to $v$, then the distance from $v$ to $u$
along $\pistar(s,t)$ is at most $4 \cdot 2^{\cutoff(v)}$.
\end{claim}

\begin{proof}
Say, for contradiction, that this was not the case.
Then, since $u$ belongs to $v$, there must be some 
vertex $z \in \ballone(v) \bigcap \ballone(u)$.
But by definition of $\ballone$ we have $\level(v,z) \leq \cutoff(v)$
and $\level(u,z) \leq \cutoff(u) \leq \cutoff(v)$, so by definition of $\level$
there is a path from $u$ to $v$ through $z$ of weight at most $4 \cdot 2^{\cutoff(v)}$,
which contradicts $\pistar(s,t)$ being the shortest path in $\gstar$.
\end{proof}

Now, given any $\tau$-heavy edge $(x,y)$, let the dominant endpoint of 
$(x,y)$ be the endpoint that comes earlier in $P$:
so in particular, if $x$ is dominant then $\cutoff(x) \geq \cutoff(y)$.
We say that a $\tau$-heavy edge $(x,y)$ belongs to independent vertex $v$ if
the dominant endpoint of $(x,y)$ belongs to $v$.
Note that every $\tau$-heavy edge in $\pistar(s,t)$
belongs to at least one independent vertex.
We now show that if $v$ is independent, 
then the total weight of $\tau$-heavy edges that belong to $v$
is at most $14 \cdot 2^{\cutoff(v)}$. 
Otherwise (for contradiction), at least half that weight
would come before or after $v$, so w.l.o.g, let us say that there is at
least $7 \cdot 2^{\cutoff(v)}$ weight of edges that belong to $v$ 
and come after $v$ on $\pistar(s,t)$ -- i.e., are closer to $t$ than $v$ is.
Let us consider the $\tau$-heavy edge $(x,y)$ that belongs to $v$ and is closest to $t$,
and say that $y$ is closer to $t$ than $x$.
Note that the distance from $v$ to $y$ along $\pistar(s,t)$ is at least 
$7 \cdot 2^{\cutoff(v)}$, so by Claim \ref{claim:belonging-distance}, $y$ cannot
belong to $v$. Thus, since edge $(x,y)$ belongs to $v$, $x$ must
be the dominant endpoint and belong to $v$. But 
if $x$ belongs to $v$ and dominates $y$ then $\cutoff(y) \leq \cutoff(x) \leq \cutoff(v)$,
and the only way edge $(x,y)$ could be $\tau$-heavy is if $\level(x,y) \leq \cutoff(v)$,
so $w(x,y) \leq 2 \cdot 2^{\cutoff(v)}$, so the distance from $v$ to $x$ along $\pistar(s,t)$
is still at least $7 \cdot 2^{\cutoff(v)} - 2 \cdot 2^{\cutoff(v)} = 5 \cdot 2^{\cutoff(v)}$,
which again contradicts Claim \ref{claim:belonging-distance}. Thus, the total
weight of heavy edges on $\pistar(s,t)$ is at most $14\sum_{v \in \indep} 2^{\cutoff(v)}$,
which combined with Equation \ref{eq:sum-cutoffs} proves Equation \ref{eq:dstar-upper},
which proves the lemma.
\end{proof}

Maintaining the threshold graph $\gtau$ as $G$ changes is easy, because the only hard part 
is maintaining the connected components $C$ in $\ghtau$,
and dynamic connectivity in undirected graphs is a well-studied problem
that admits deterministic $O(\log^2(n))$ update time. The details
are essentially identical to those for unweighted graphs (see Lemmas 4.5 and 4.6 in \cite{BernsteinC16}), so we omit the proofs of the following two lemmas.

\begin{lemma}
\label{lem:gtau-maintain}
Given a graph $G$ subject to a decremental update sequence,
and a positive integer threshold $\tau$,
we can maintain the graph $\gtau$ in \emph{total} time
$O(m\log^2(n))$.
Moreover, the total number of edges of weight $1/2$ ever inserted into
the graph is $O(n\log(n))$,
and the total number of edges $(u,v)$ that have $\level(u,v) = i$ 
when inserted is at most $n \tau 2^i$.
\end{lemma}

\begin{lemma}
\label{lem:gtau-decremental}
Given a graph $G$ subject to a decremental update sequence,
a positive integer threshold $\tau$,
and a pair of vertices $u,v$ in $G$,
the distance $\dtau(u,v)$ in $\gtau$
never decreases as the graph $G$ changes.
\end{lemma}

\section{An Extension of the Even and Shiloach Algorithm}

In the paper of Bernstein and Chechik for unweighted graphs \cite{BernsteinC16},
the threshold graph ended up being sparse, 
and so it was easy to efficiently compute distances within it
using the standard Even and Shiloach algorithm with total update time $O(mn)$ \cite{EvenS81}.
In our paper, however, the threshold graph is only sparse
with respect to low-weight edges; 
that is, because heaviness and lightness is defined in terms of edge weights,
the threshold graph can have many edge of high weight. We
now present a modification of the Even and Shiloach algorithm that can deal
more efficiently with edges of high weight,
which makes it perfectly suited to our threshold graph,
which has mostly high-weight edges. 
Like the standard Even and Shiloach algorithm,
our extension runs up to a certain depth bound $d$.

\begin{definition}
\label{def:bound}
Given any number $d$, the function $\bound{d}(x)$ is equal to $x$
if $x \leq d$, and to $\infty$ otherwise.
\end{definition}

Note that 
to apply to the threshold-graph $\gtau$, the
algorithm must be able to handle insertions,
as long as they do not change distances.
(See Lemmas \ref{lem:gtau-maintain}, \ref{lem:gtau-decremental}).
To this end we need the following definition:

\begin{definition}
\label{def:matha}
Let $G = (V,E)$ be a graph with positive weights $\geq 1$,
subject a dynamic update sequence of 
insertions, deletions, and weight-increases. In such a context, let
$\matha$ contain the set of ALL edges $(u,v)$ to appear
in $G$ at any point during the update sequence ($\matha$ for all): if 
an edge $(u,v)$ is deleted and then inserted again,
we consider this as two separate edges in $\matha$.  
For every edge $(u,v) \in \matha$, let
$\oweight(u,v)$ (o for original) be the weight of $(u,v)$ when it first appears in $\matha$ 
(this might be in the original graph itself),
and let $\olevel(u,v) = \floor{\log(\oweight(u,v))}$.
\end{definition} 

\begin{lemma}
Let $G = (V,E)$ be an undirected graph with positive integer edge weights,
$s$ a fixed source, and $d \geq 1$ a depth bound. Say that $G$ is subject
to a dynamic sequence of edge insertions, deletions, and weight increases,
with the property that distances in $G$ never decrease as a result of an update.
Then, there exists an algorithm $\wses(G,s,d)$ (stands for weight-sensitive Even and Shiloach)
that maintains approximate distances $\dist'(s,v)$
with $\dist(s,v) \leq \dist'(s,v) \leq \bound{d}(\dist(s,v)) \oeps$,
and has total update time 
$O(nd + \sum_{(u,v) \in \matha} \frac{d}{\eps \oweight(u,v)})$.  
\end{lemma}

\begin{proof}
Throughout the proof of this lemma, we will often rely on the following function:

\begin{definition} 
\label{def:additive-rounding}
Given any positive real numbers $\beta$ and  $x$, let 
$\round{\beta}{x}$ be the smallest number $y > x$ that is an integer multiple of $\beta$.
\end{definition}
 
Our algorithm $\wses$ follows the basic procedure of the classic Even and Shiloach algorithm
\cite{EvenS81} (denoted $\es$), with a few modifications. 
For this reason we can describe many of the guarantees
of $\es$ without looking under the hood for exactly how they are implemented,
since we keep exactly the same implementation. We only modify the classic algorithm in a few key points.
Recall the definition of $\matha$ and $\oweight(u,v)$ from Definition \ref{def:matha}. 
Since weights only increase, we always have $w(u,v) \geq \oweight(u,v)$.

The $\es$ algorithm maintains for each vertex $v$ a distance label $\del(v)$ with the property
that $\dist(s,v) = \del(v)$. In particular, it maintains the invariant that $\del(s) = 0$, and 
that if $(u,v) \in E$, then $\del(v) \leq \del(u) + w(u,v)$. Our algorithm
will also maintain a label $\del(v)$ for every $v$, but since we only need
an approximation, we guarantee a weaker invariant:

\begin{invariant} ({\bf Approximation})
We always have $\del(s) = 0$. If $p$ is the parent of $v$ in the tree then 
$\del(v) \geq \del(p) + w(p,v)$.
Also, for all $(u,v) \in E$, $\del(v) \leq \del(u) + w(u,v) + \eps \oweight(u,v)$.
\end{invariant}

\noindent {\bf Approximation Analysis:}
We now show that as long as the Approximation Invariant is preserved, we always have 
$\dist(s,v) \leq \del(v) \leq (1 + \eps) \dist(s,v)$.
The fact that $\del(v) \geq \dist(s,v)$ follows trivially for the same reason as in classic $\es$.
To prove the other side, consider the shortest path $\pi(s,v)$ in $G$.
Let the vertices in $\pi(s,v)$ be $s_0, s_1, ..., s_q = v$. We show by induction
that $\del(s_i) \leq (1 + \eps) \dist(s,s_i)$. The claim is trivially true for $i=0$
because $\del(s) = 0$ at all times. Now, say the claim is true for $s_i$,
and note that $\dist(s,s_{i+1}) = \dist(s,s_i) + w(s_i, s_{i+1})$.
Now, by the Approximation Invariant, we have 
$\del(s_{i+1}) \leq \del(s_i) + w(s_i, s_{i+1}) + \eps \oweight(s_i,s_{i+1}) \leq 
\del(s_i) + (1 + \eps) w(s_i, s_{i+1})$ 
(the last step follows from $w(u,v) \geq \oweight(u,v)$).
But then by the induction hypothesis we have 
$\del(s_{i+1}) \leq (1 + \eps)(\dist(s,s_i) + w(s_i, s_{i+1})) = (1 + \eps) \dist(s,s_{i+1})$.

We must now show how to efficiently maintain the Approximation Invariant. The classic $\es$ algorithm
does $O(1)$ time per update for basic setup, and otherwise all the work comes 
from charging constant time operations to various edges $(u,v)$ and vertices $v$. 
The update time invariant in classic $\es$ is as follows: we
only charge vertex $v$ when $\del(v)$ increases, 
and we only charge edge $(u,v)$ when $\del(u)$ or $\del(v)$ increases.
Since weights are positive integers, every $\del(u)$ must increase by at least $1$,
and since we only go up to distance $d$, each $\del(v)$ increases at most $d$ times, 
which yields $O(md)$. To improve upon this, our algorithm 
maintains a slightly different update time invariant

\begin{invariant} ({\bf Update-Time})
A vertex $v$ is only charged when $\del(v)$ increases. An edge $(u,v)$ is
only charged when either $\round{\eps \oweight(u,v)}{\del(u)}$ or 
$\round{\eps \oweight(u,v)}{\del(v)}$ increases.
\end{invariant}

Since for any $\beta$, $\round{\beta}{\del(v)}$ can 
increase at most $d / \beta$ before it exceeds $d$, 
it is clear that the Update-Time Invariant
guarantees the total update bound of the lemma.
\\
{\bf Maintaining the Two Invariants:}
Now, let us recall how the standard $\es$ algorithm works.
This algorithm only handles deletions, so that is what
we focus on first: we later show how to extend
our algorithm to handle weight-increases and insertions that
do not change distances. 
We will emphasize the parts of classic $\es$ that we will later change in our modified algorithm. 
First off, each vertex $v$ always maintains label information
about its neighbors: \emph{so in particular, for each edge $(u,v)$,
vertex $v$ stores a copy $\del_v(u) = \del(u)$} (\firstdif). 
The algorithm maintains a shortest path tree $T$ from the root $s$.
We say an edge $(u,v)$ is a tree edge if $(u,v) \in T$, and a non-tree edge if $(u,v) \in E \sm T$.
Whenever a non-tree edge is deleted, shortest distances do not change, so the algorithm does not have to do much.
Now, consider the deletion of a tree edge $(u,v)$ where $u$ is the parent of $v$ in $T$. The algorithm checks in constant time, by cleverly storing local label information $\del_v(z)$, if $v$ has another edge to a vertex $z$ with $\del(v) = \del_v(z) + w(z,v)$.
If yes, $v$ can be reconnected without a label increase,
the edge $(z,v)$ is added to the tree and all distance labels remain the same so we are done.
If not, label $\del(v)$ must increase, which in turn affects of the children of $v$, so
 \emph{The algorithm then examines all the children of $v$ and checks if they can be attached to the tree without increasing their distance in a similar manner.}
(\seconddif)
Some vertices will fail to be reattached with the same label, 
so the algorithm keeps a list of all vertices that need to be re-attached at a higher label.
Each time the algorithm discovers that a label $\del(v)$ must increase, it increases it by 1 (which is the most optimistic new label it can get) and returns it to the needs-reattachment list in at attempt to reattach $v$ with this higher label.
The algorithm examines the vertices in the list in an increasing order of their label.
Note that the label of a vertex may be increased many times as a result of an update.
\emph{Whenever the label $\del(v)$ of a vertex $v$ changes, the algorithm
must adjust $\del_u(v)$ for every edge $(u,v)$ in the graph.} (\thirddif.)

Note that the algorithm above fails to satisfy the Update-Time Invariant,
because although a vertex $v$ is only charged when $\del(v)$ increases 
(otherwise the search stops at $v$ and can be charged to the edge that led to $v$),
edge $(u,v)$ is charged whenever $\del(u)$ or $\del(v)$ increase even just by $1$.
In particular, the invariant is not preserved because of the work
in \seconddif\ and \thirddif.
The changes we implement are simple. We start with \firstdif: 
each vertex $v$ will still store local label information $\del_v(u)$, but now
it might be slightly inaccurate. In particular, we always have

\begin{invariant} ({\bf Local-Information})
$\del(u) \leq \del_v(u) \leq \round{\eps \oweight(u,v)}{\del(u)}$.
\end{invariant}

A vertex $v$ will now choose its parent based on the slightly inaccurate local labels $\del_u(v)$
instead of the true label $\del(u)$. Let $\parent(v)$ be the parent of $v$ in the tree.

\begin{invariant} ({\bf Parent-Choice})
\\
$\parent(v) = \mbox{argmax}_{u \in \neighbors(v)} \{ \del_v(u) + w(u,v) \}$ 
and $\del(v) = \del_v(\parent(v)) + w(\parent(v),v)$.
\end{invariant}

The Parent-Choice and Local-Information invariants combined 
clearly guarantee the Approximation Invariant.
We must now show how to maintain these efficiently. 
This leads us to \thirddif: when $\del(v)$ changes, we only update
$\del_u(v)$ for every vertex $u$ when
$\round{\eps \oweight(u,v)}{\del(v)}$ increases. 
This ensures that we
satisfy the Update-Time Invariant for this step of the algorithm,
while still ensuring that all local information adheres to the
Local-Information Invariant.

Note that to implement \thirddif\ in the above paragraph
while maintaining the Update-Time invariant,
we need a data structure for every vertex $v$,
such that whenever $\del(v)$ increases,
the data structure returns all vertices $u$ for which 
$\round{\eps \oweight(u,v)}{\del(v)}$ increased.
This is easy to do because $\oweight(u,v)$ is completely fixed for each edge,
i.e. it does not change as the graph changes. 
Each vertex $v$ will have $d$ buckets, where bucket $i$
contains all edges $(u,v)$ such that
$\round{\eps \oweight(u,v)}{i-1} < \round{\eps \oweight(u,v)}{i}$.
This initialization requires $O(nd)$ time to create $d$ buckets
for vertex, and then each edge $(u,v)$ is placed in $d/\oweight(u,v)$
buckets for vertices $u$ and $v$, leading to total time 
$O(nd + \sum_{(u,v) \in \matha} \frac{d}{\eps \oweight(u,v)})$,
as desired. Once the buckets are initialized,
whenever $\del(v)$ increases from $i$ to $j$,
$\del(v)$ simply notifies all edges in buckets $i+1, i+2, ..., j$
 
The last violation of the Update-Time Invariant was in \seconddif. 
When $\del(v)$ changes, instead of processing all children $u$,
we only need to process those for which $\del_u(v)$ changes. But in our modified
algorithm this only occurs when $\round{\eps \oweight(u,v)}{\del(v)}$ increases,
in keeping with the Update-Time Invariant.

We have thus shown how to modify classical $\es$ in a way that maintains
the Approximation Invariant and the Update-Time Invariant while processing deletions.
Increase-Weight$(u,v)$ is processed in the same way: if $(u,v)$
was not a tree edge we do nothing, and otherwise
if $u$ was the parent of $v$, then some invariants will be violated for $v$,
and we fix these in exactly the same way as for a deletion.

We now turn to processing insertions. In classical $\es$, it is
trivial to process an insertion of $(u,v)$ that does decrease distances
because the distance labels do not change, so we just spend $O(1)$
time fixing update $\del_v(u)$ and $\del_u(v)$.
In our algorithm, however, a strange difficulty arises:
because our labels are approximate,
even if the insertion of $(u,v)$ does not decrease distances,
it might nonetheless violate the Approximation Invariant: say that $\dist(u) = \dist(v) = 1000$
but $\del(v) = 1000(1 + \eps)$ while $\del(u) = 1000$ and we insert and edge of weight $1$.
In this case we  cannot afford to decrease $\del(v)$ because
our update time analysis relied on non-decreasing labels.
But note that although $v$ now violates the Approximation Invariant, $\del(v)$
is still a $(1 + \eps)$ approximation to $\dist(s,v)$, because $\dist(s,v)$ never decreases.
We thus rely on the idea used in the "Monotone" $\es$ tree of 
Henzinger et al. \cite{HenzingerKN13}: we simply never decrease distance labels.
More formally, our distance labels will satisfy the following \emph{relaxed} Approximation Invariant:
after every update, for every vertex $v$, either $\del(v)$ satisfies the standard Approximation
Invariant, or $\del(v) = \del^{\old}(v)$, where $\del^{\old}(v)$ was the distance
label before the update. By an induction on time,
it is easy to see that with this relaxed Approximation Invariant,
we still always have that $\del(v) \leq (1+\eps) \dist(s,v)$.
If $\del(v)$ satisfies the approximation invariant, the proof
is the same as in the Approximation Analysis above. Otherwise,
we have $\del(v) = \del^{\old}(v)$, in which case by our time induction we have
$\del^{\old}(v) \leq (1+\eps)\dist^{\old}(s,v)$, and since our updates
are guaranteed not to decrease distances, we have 
$\dist^{\old}(s,v) \leq \dist(s,v)$ and we are done. 
\end{proof}

\section{The decremental SSSP algorithm}

We now put together all our ingredients to proving the main Theorem \ref{thm:main}.

\begin{lemma}
\label{lem:wses-gtau}
For any threshold $\tau$, and depth bound $d = n \eps^{-1} \tau^{-1}$, 
we can maintain approximate
distances $\dptau(s,v)$ in total time 
$O(m\log^2(n) + n^2\log(n)\eps^{-2})$,
with the property that $\dptau(s,v) \leq (1+\eps)\bound{d}(\dtau(s,v)) + d\eps$
\end{lemma}

\begin{proof}
First we use Lemma \ref{lem:gtau-maintain}
to maintain the threshold graph $\gtau$
in $O(m\log^2(n))$ total update time.
Then, setting $\beta = \frac{d\eps}{2n}$, 
for each edge $(u,v) \in \gtau$,
we round the edge weight $w(u,v)$ up to the nearest multiple
of $\beta$.
It is easy to see that since every shortest path
in the graph contains at most $n$ edges,
this weight change incurs an additive error of at most $n\beta = d\eps/2$.
Since all edge weights are now multiple of $\beta$, 
we can divide all edge weights by factor of $\beta$ without changing shortest paths,
which results in a scaled graph $\gtaus$ with integral
weights for which we have to maintain shortest distances
up to depth $\ds = d/\beta = 2n \eps^{-1}$.
We now run $\wses(\gtaus,s,\ds)$ in the scaled graph,
and then scale the resulting distances back up by $\beta$.
Since $\wses$ is a $\oeps$ approximation, we get:
$\dptau(s,v) \leq (1+\eps)\bound{d}(\dtau(s,v) + d\eps/2) \leq (1+\eps)\bound{d}\dtau(s,v) + d\eps$. 

Recall that the total update time of $\wses(\gtaus, s,\ds)$ is
$O(n\ds + \sum_{(u,v) \in \matha} \frac{\ds}{\eps \oweight(u,v)})$.
We know that $O(n\ds) = O(n^2 \eps^{-1})$.
Now, $\gtaus$ contains two types of edges: 
edges between a vertex $v \in V$ and a component vertex $c \in \vtau \sm C$,
and $\tau$-light edges from the original graph. 
By Lemma \ref{lem:gtau-maintain}, the total number of component
edges ever inserted is $O(n\log(n))$, and because of our scaling
they all have weight at least $1$, 
so their total contribution to the sum is at most
$n\log(n) \eps^{-1} \cdot \ds = O(n^2\log(n)\eps^{-2})$.
For edges of the original graph, let us look at the contribution
of all edges $(u,v)$ with $\olevel(u,v) = i$. By
Lemmas \ref{lem:gtau-sparse} and \ref{lem:gtau-maintain},
the total number of such edges is at most $O(n \tau 2^i)$.
Each such edge has $\oweight(u,v) \geq 2^i$, which
has been scaled down by $\beta$ in $\gtaus$,
and so contributes
$\frac{\ds\beta}{\eps 2^i} = \frac{d}{\eps 2^i}$ to the sum. Thus
in total we have that each level $i$
contributes a total of $O(n \tau d \eps^{-1}) = O(n^2 \eps^{-2})$.

We complete the proof by arguing that there are only $\log(n)$ contributing levels.
For simplicity, let us look at the graph $\gtau$ before scaling.
First off, edges of weight more than $d = n\eps^{-1}/\tau$
can be ignored, since we do not care about distances greater than $d$
(we are working with $\bound{d}$).
Secondly, from the definition of $\cutoff(v)$, an edge $(u,v)$
of weight less than $1/(2\tau)$ will always be heavy and so never
appear in $\gtau$:
if $2\tau w(u,v) \leq 2^{\level(u,v)}\tau < 1$,
then $\cutoff(u)$ and $\cutoff(v)$ are always at least as large as $\level(u,v)$,
so $(u,v)$ will remain heavy. Thus, the only edges that appear in our graph
have weight between $1/(2\tau)$ and $n/(\eps\tau)$, which corresponds
to $\log(n/\eps) = O(\log(n))$ different levels of edge weights. 
\end{proof}

\begin{lemma}
\label{lem:distance-d}
For any distance bound $d$, with total update time $O(m\log^2(n) + n^2\log(n)\eps^{-2})$
we can maintain approximate distances $\dist_d(s,v)$ to all vertices $v$ such that \\
$\dist(s,v) \leq \dist_d(s,v) \leq \bound{d}(\dist(s,v)) + 15\eps d$
\end{lemma}

\begin{proof}
Set $\tau = n / (\eps d)$. By Lemma \ref{lem:wses-gtau}, in total update time 
$O(m\log^2(n) + n^2\log(n)\eps^{-2})$ we can maintain values $\dptau(s,v)$
with $\dtau(s,v) \leq \dptau(s,v) \leq \oeps \bound{d}(\dtau(s,v)) + \eps d$.
Now, By Lemma \ref{lem:gtau-approx} we have 
$\dist(s,v) - 14n/\tau \leq \dtau(s,v) \leq \dist(s,v)$.
Plugging in our value $\tau = n/(\eps d)$ we get 
$\dist(s,v) - 14\eps d \leq \dtau(s,v) \leq \dist(s,v)$.
Thus, we have 
$\dist(s,v) - 14\eps d \leq \dptau(s,v) \leq \bound{d}(\dist(s,v)) + \eps d$.
If we return $\dist_d(s,v) = \dptau(s,v) + 14\eps d$, we get
the bound in the lemma.
\end{proof}

It is now easy to prove our main Theorem \ref{thm:main}. Observe that the maximum
distance we could possibly see is $nW$. Thus, for each
$i = 0,1,2,3...\ceil{\log(nW)}$, we use Lemma \ref{lem:distance-d} to maintain $\del_i(v) = \dist_{2^i}(s,v)$. By Lemma \ref{lem:distance-d}, the total update time is the desired
$O(m\log^2(n)\log(nW) + n^2\log(n)\log(nW)\eps^{-2})$.
We then output as our final answer:
 $\dist'(v) = \min_i \{ \del_i(v) \}$.
Since for each $i$ we have $\del_i \geq \dist(v)$, we have $\dist'(v) \geq \dist(s,v)$.
We now need to show that $\dist'(v) \leq (1 + O(\eps))\dist(s,v)$.
It is not hard to see that for $i = \ceil{\log(\dist(s,v))}$, we end up with
$\del_i(v) \leq \dist(s,v) + 30\eps \dist(s,v)$, as desired.

\section{Conclusions}
In this paper we presented the first \emph{deterministic} decremental SSSP algorithm for
\emph{weighted} undirected graphs that goes beyond the Even-Shiloach bound of $O(mn)$ total update time.
Previously such a result was only known for unweighted undirected graphs. The two main open questions
are further improving this total update time, 
and going beyond $O(mn)$ deterministically for \emph{directed} graphs.
Finally, we recall from the introduction that all existing deterministic $o(mn)$ 
results including ours are only able to return approximate shortest distances,
not the paths themselves. 

\newpage
\bibliography{shiri_icalp2017}

\begin{thebibliography}{10}

\bibitem{AbrahamCT14}
Ittai Abraham, Shiri Chechik, and Kunal Talwar.
\newblock Fully dynamic all-pairs shortest paths: Breaking the o(n) barrier.
\newblock In {\em Approximation, Randomization, and Combinatorial Optimization.
  Algorithms and Techniques, {APPROX/RANDOM} 2014, September 4-6, 2014,
  Barcelona, Spain}, pages 1--16, 2014.

\bibitem{BaHaSe03}
Surender Baswana, Ramesh Hariharan, and Sandeep Sen.
\newblock Maintaining all-pairs approximate shortest paths under deletion of
  edges.
\newblock In {\em Proceedings of the Fourteenth Annual ACM-SIAM Symposium on
  Discrete Algorithms}, SODA, pages 394--403. Society for Industrial and
  Applied Mathematics, 2003.

\bibitem{Be09}
Aaron Bernstein.
\newblock Fully dynamic (2 + epsilon) approximate all-pairs shortest paths with
  fast query and close to linear update time.
\newblock In {\em Proceedings of the 50th Annual IEEE Symposium on Foundations
  of Computer Science}, FOCS, pages 693--702, 2009.

\bibitem{BernsteinC16}
Aaron Bernstein and Shiri Chechik.
\newblock Deterministic decremental single source shortest paths: beyond the
  o(mn) bound.
\newblock In {\em Proceedings of the 48th Annual ACM Symposium on Theory of
  Computing (STOC)}, pages 389--397, 2016.

\bibitem{BernsteinC17}
Aaron Bernstein and Shiri Chechik.
\newblock Deterministic partially dynamic single source shortest paths for
  sparse graphs.
\newblock In {\em Proceedings of the Twenty-Eighth Annual {ACM-SIAM} Symposium
  on Discrete Algorithms, {SODA} 2017, Barcelona, Spain, Hotel Porta Fira,
  January 16-19}, pages 453--469, 2017.

\bibitem{BernsteinR11}
Aaron Bernstein and Liam Roditty.
\newblock Improved dynamic algorithms for maintaining approximate shortest
  paths under deletions.
\newblock In {\em Proceedings of the Twenty-Second Annual {ACM-SIAM} Symposium
  on Discrete Algorithms}, SODA, pages 1355--1365, 2011.

\bibitem{Di06}
Yefim Dinitz.
\newblock Dinitz' algorithm: The original version and even's version.
\newblock In {\em Theoretical Computer Science, Essays in Memory of Shimon
  Even}, pages 218--240, 2006.

\bibitem{EvenS81}
Shimon Even and Yossi Shiloach.
\newblock An on-line edge-deletion problem.
\newblock {\em Journal of the ACM}, 28(1):1--4, 1981.

\bibitem{HenzingerK95FOCS}
Monika Henzinger and Valerie King.
\newblock Fully dynamic biconnectivity and transitive closure.
\newblock In {\em 36th Annual Symposium on Foundations of Computer Science,
  Milwaukee, Wisconsin, 23-25 October 1995}, pages 664--672, 1995.

\bibitem{HenzingerKN13}
Monika Henzinger, Sebastian Krinninger, and Danupon Nanongkai.
\newblock Dynamic approximate all-pairs shortest paths: Breaking the o(mn)
  barrier and derandomization.
\newblock In {\em Proceedings of the 54th Annual Symposium on Foundations of
  Computer Science}, FOCS, pages 538--547, 2013.

\bibitem{HenzingerKNFOCS14}
Monika Henzinger, Sebastian Krinninger, and Danupon Nanongkai.
\newblock Decremental single-source shortest paths on undirected graphs in
  near-linear total update time.
\newblock In {\em Proceedings of the 55th Annual Symposium on Foundations of
  Computer Science}, FOCS, pages 146--155, 2014.

\bibitem{HenzingerKNSTOC14}
Monika Henzinger, Sebastian Krinninger, and Danupon Nanongkai.
\newblock Sublinear-time decremental algorithms for single-source reachability
  and shortest paths on directed graphs.
\newblock In {\em Proceedings of the 46th Annual ACM Symposium on Theory of
  Computing (STOC)}, pages 674--683, 2014.

\bibitem{HeKrNa14}
Monika Henzinger, Sebastian Krinninger, and Danupon Nanongkai.
\newblock A subquadratic-time algorithm for decremental single-source shortest
  paths.
\newblock In {\em Proceedings of the Twenty-Fifth Annual ACM-SIAM Symposium on
  Discrete Algorithms}, SODA, pages 1053--1072, 2014.

\bibitem{HenzingerKN15}
Monika Henzinger, Sebastian Krinninger, and Danupon Nanongkai.
\newblock Improved algorithms for decremental single-source reachability on
  directed graphs.
\newblock In {\em Proceedings of the 42nd International Colloquium, {ICALP}},
  pages 725--736, 2015.

\bibitem{HenzingerKNS15}
Monika Henzinger, Sebastian Krinninger, Danupon Nanongkai, and Thatchaphol
  Saranurak.
\newblock Unifying and strengthening hardness for dynamic problems via the
  online matrix-vector multiplication conjecture.
\newblock In {\em Proceedings of the Forty-Seventh Annual ACM on Symposium on
  Theory of Computing (STOC)}, pages 21--30, 2015.

\bibitem{HeKi01}
Monika~Rauch Henzinger and Valerie King.
\newblock Maintaining minimum spanning forests in dynamic graphs.
\newblock {\em {SIAM} J. Comput.}, 31(2):364--374, 2001.

\bibitem{King99}
Valerie King.
\newblock Fully dynamic algorithms for maintaining all-pairs shortest paths and
  transitive closure in digraphs.
\newblock In {\em Proceedings of the 40th Annual Symposium on Foundations of
  Computer Science}, FOCS, pages 81--91, 1999.

\bibitem{Madry10}
Aleksander Madry.
\newblock Faster approximation schemes for fractional multicommodity flow
  problems via dynamic graph algorithms.
\newblock In {\em Proceedings of the 42nd {ACM} Symposium on Theory of
  Computing, {STOC} 2010, Cambridge, Massachusetts, USA, 5-8 June 2010}, pages
  121--130, 2010.

\bibitem{RoZw04b}
Liam Roditty and Uri Zwick.
\newblock On dynamic shortest paths problems.
\newblock {\em Algorithmica}, 61(2):389--401, 2011.

\bibitem{RoZw12}
Liam Roditty and Uri Zwick.
\newblock Dynamic approximate all-pairs shortest paths in undirected graphs.
\newblock {\em SIAM J. Comput.}, 41(3):670--683, 2012.

\end{thebibliography}

\end{document}